\documentclass[letterpaper, 10 pt, conference]{ieeeconf}  
\usepackage{style}

\makeatletter
\def\ps@IEEEtitlepagestyle{
  \def\@oddhead{\hbox{}\normalfont\scriptsize\leftmark\hfil}\relax
  \def\@oddfoot{\mycopyrightnotice}
  \def\@evenfoot{}
}
\makeatother

\def\mycopyrightnotice{
  {\footnotesize
  \begin{boxedminipage}{\textwidth}
  \centering
  © 2026 IEEE. Personal use of this material is permitted. Permission from IEEE must be obtained for all other uses, in any current or future media, including reprinting/republishing this material for advertising or promotional purposes, creating new collective works, for resale or redistribution to servers or lists, or reuse of any copyrighted component of this work in other works. 
  \end{boxedminipage}
  }
}

\def\BibTeX{{\rm B\kern-.05em{\sc i\kern-.025em b}\kern-.08em
    T\kern-.1667em\lower.7ex\hbox{E}\kern-.125emX}}
\title{\LARGE \bf
Statistical Contraction for Chance-Constrained Trajectory Optimization of Non-Gaussian Stochastic Systems}

\author{Rihan Aaron D'Silva and Hiroyasu Tsukamoto
\thanks{This work benefited from technical discussions within DARPA’s Safe and Assured Foundation Robots for Open eNvironments (SAFRON) Program, under contract number HR0011-25-3-0331. The authors are with the Department of Aerospace Engineering, The Grainger College of Engineering, University of Illinois Urbana-Champaign, Urbana, Illinois  61801, {\tt\small rihanad2@illinois.edu, hiroyasu@illinois.edu}.}%
}

\begin{document}
\maketitle

\thispagestyle{IEEEtitlepagestyle}
\pagestyle{headings}


\begin{abstract}
We present a distribution-free approach to robust trajectory optimization and control of discrete-time, nonlinear, and non-Gaussian stochastic systems, with closed-loop guarantees on chance constraint satisfaction. Our framework employs conformal inference to generate coverage-based confidence sets for the closed-loop dynamics around arbitrary reference trajectories under uncertainty. It thereby constructs a joint nonconformity score to quantify both the validity of contraction (i.e., incremental stability) conditions and the impact of external stochastic disturbance on the closed-loop dynamics, without any distributional assumptions. Via appropriate constraint tightening, chance constraints can be reformulated into tractable, statistically valid deterministic constraints on the reference trajectories. This enables a formal pathway to certify the performance of learning-based motion planners and controllers, such as those with neural contraction metrics, in safety-critical real-world applications.
Notably, our statistical guarantees are non-diverging and can be computed with finite samples of the underlying uncertainty, without overly conservative structural priors. We demonstrate our approach in motion planning problems for designing safe, dynamically feasible trajectories in both numerical simulations and hardware experiments.


\end{abstract}
\section{Introduction}
\label{sec_intro}
Generating interpretable, high-performance motion plans for dynamical systems under unstructured uncertainties is central to their deployment in safety-critical applications. As the scope and scale of the tasks entrusted to these systems have expanded significantly in recent years, such planning problems have become increasingly nonlinear and stochastic, with the underlying uncertainties often non-Gaussian. Consequently, data-driven and learning-based approaches have been widely pursued as promising alternatives to classical formulations. As the gap between their empirical performance and theoretical guarantees continues to widen, we take a step back and establish a statistically rigorous, closed-loop perspective to address this challenge.


Among classical model-based approaches, chance-constrained programming~\cite{chance_constraint_original} provides a natural framework for handling system uncertainties in motion planning. It operates through relaxing hard constraints into probabilistic counterparts that hold with high confidence, yielding less conservative solutions. Since nominal chance constraints are generally nonconvex and often computationally intractable, they may be reformulated into tractable deterministic problems by introducing distributional assumptions or approximations of underlying uncertainty. For example, existing methods consider additive Gaussian noise~\cite{Covesteering, BlackmoreOnoChanceConstrained}, which allows propagation of the uncertainty through linear dynamics in closed form, leading to deterministic relaxations of chance constraints.

Alternatively, distributionally robust approaches seek constraint satisfaction across a family of probability distributions belonging to an ambiguity set. Such a set is commonly defined via moment specifications~\cite{calafiore2006distributionally} or via distributional metrics~\cite{DixitDRMPC, CoulsonLygerosDorfler2021}. However, extending these ideas beyond linear settings becomes challenging, since uncertainty or ambiguity sets cannot, in general, be propagated analytically through nonlinear dynamics. Several ideas have been proposed to address this issue, e.g., through linearization~\cite{LewCCSCP ,CutiousMPCGP}, polynomial chaos expansions~\cite{NakkaCCnonlinearS}, or learning-based methods~\cite{gp_learning,sun2021uncertaintyaware}. These methods propagate statistical information through approximated dynamics, albeit at the cost of strict probabilistic guarantees. 

\begin{figure}[t]  
    \centering
    \vspace{0.5em}
    \includegraphics[width=\linewidth]{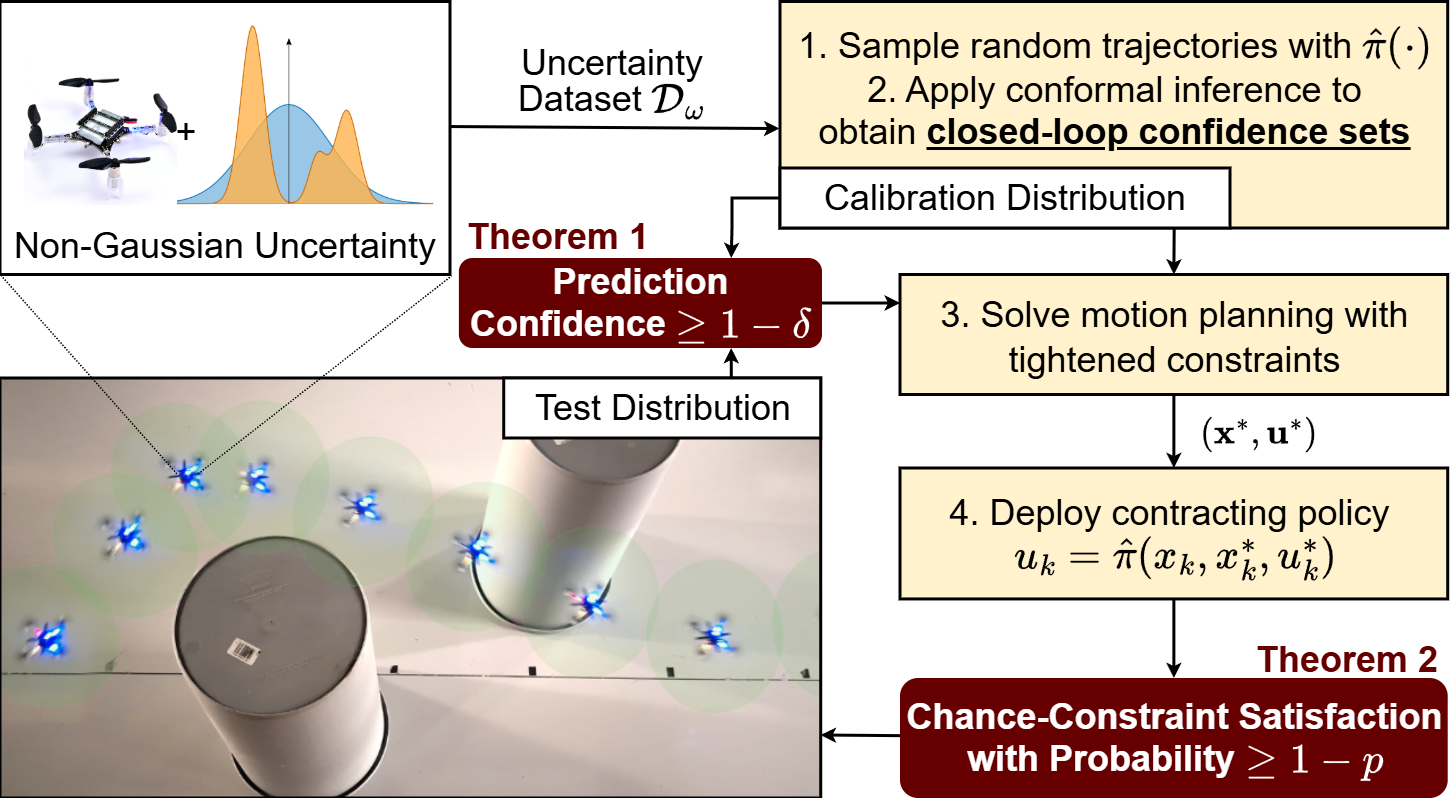}  
    \caption{Illustration of our proposed approach. Given a finite dataset \(\mathcal{D}_w\) of non-Gaussian disturbance samples, from predictions of the control contraction metrics \(\mathbf{\hat{M}}\), tracking policy \(\hat{\pi}(\cdot)\) for contraction rate \(\lambda\), we infer high probability chance constraint satisfaction guarantees for safe motion planning.
    Code : \myhref{https://github.com/Rihan24/SCC-TrajOpt}{https://github.com/Rihan24/SCC-TrajOpt}.
    } 
    

    \label{fig:intro}
    \vspace{-2em}
\end{figure}

Sampling-based approaches address this challenge from a different angle, pursuing formal a priori guarantees achievable purely through data. The scenario approach bypasses explicit distributional modeling by enforcing constraints directly on a finite set of sampled realizations or scenarios~\cite{Scenario,schildbach2014scenario}. It enables statistically valid inference of probabilistic certificates on constraint satisfaction. Sample average approximation methods aim to, instead, approximate chance constraints using the empirical distribution over the samples~\cite{Hakobyan2019SAA}.
Another line of work is conformal inference, also known as conformal prediction (CP)~\cite{conformal1, conformalbook,WCP_Barberetal}, which has recently gained attention in control theory as a powerful, yet relatively light-weight tool for distribution-free uncertainty quantification associated with point or time-series predictions. Given a predictor and a negatively oriented score function to measure the predictor's performance, it exploits the rank statistics of the score to construct a high-confidence bound on the prediction error. CP and its variants have been employed in chance-constrained optimization~\cite{CPforChanceConstrainedProgramming} and a variety of planning and control settings~\cite{lindemann2024formal}. This includes, but is not limited to, safe planning in unknown dynamic environments~\cite{confcont2,adaptive_conformal}, deterministic robustification of model-based controllers~\cite{confcont3,2023_Chee_WPC-MPC,2024_Zhou_ACP-CBF-MPC}, uncertainty-aware planning with diffusion models~\cite{2023_Sun_CP-Diffusion}, data-driven state estimation~\cite{2023_Yang_CP-Sensor}, prediction region construction for linear systems with non-Gaussian noise~\cite{CPLinearStochastic}, conformally robust Lyapunov and barrier function synthesis~\cite{TingWeiCProbust}, and deterministic contraction-based planning and control~\cite{SWei_conformalContraction}.  

However, streamlined methods that exploit CP to alleviate distribution mismatches in closed-loop trajectory planning of nonlinear, non-Gaussian stochastic systems remain underexplored. While there are a few related efforts in contraction theory, they often operate under restrictive assumptions. For example, extreme value theory-informed collision checking~\cite{EVTforContraction} provides statistical safety with a contracting controller but only under bounded perturbations. Stochastic contraction-based predictive control~\cite {KohlerMelliner} provides probabilistic reachability but requires an upper bound on the norm of the noise covariance.

This work investigates how the distribution-free concepts in CP can be integrated with the strong structural guarantees of contraction-based robust motion planning~\cite{slotine1998contraction,ccm,7989693,8814758,9303957, tsukamoto2021contraction}. We specifically consider discrete-time nonlinear stochastic systems, where the uncertainty may be non-Gaussian. We leverage a CP-based constraint tightening method to define, for each uncertainty sample, a joint nonconformity score that quantifies both the predicted distributional mismatch and the incremental stability violations of a learned contraction metric and its associated controller. We show that this approach offers a statistically rigorous characterization of how far the system's true trajectory may deviate from the reference trajectory only with finite data samples. This thereby offers a non-diverging, closed-loop guarantee of generating provably safe and dynamically feasible motion plans. Here, chance constraints are reformulated into their deterministic counterparts, agnostic to how the distribution is predicted or how the metric is constructed, with no strong structural priors. Unlike the scenario-based approach, the problem size does not grow with the number of data samples. We validate the performance of our approach through numerical simulations with the Dubins Car and hardware experiments with the Crazyflie drone, both under non-Gaussian uncertainty. The results indicate that it indeed achieves the closed-loop satisfaction of constraints at the desired probability level.

\textbf{Notation :} The set \(\Ibb_{\{a,b\}}\) denotes the collection of integers in \(\{a,b\}\) and \(\Ibb_{\geq 0}\) denotes the set of non-negative integers. We will use bold face letters \(\mathbf{x}=\{x_0, x_1, \ldots, x_N\}, \mathbf{u} = \{u_0, u_1, \ldots, u_{N-1}\} \) and \(\mathbf{w} = \{w_0, w_1, \ldots, w_{N-1}\}\) to denote the state, control and disturbance trajectories over a time horizon \(N\in\mathbb{N}\). 
\(\text{Quantile}_{1- \alpha}(\mathcal{P})\) denotes the \((1- \alpha)\)th quantile of the distribution \(\mathcal{P}\). For $A, B \in \mathbb{R}^{n \times n}$, we use $A\succeq B$ to denote positive semi-definiteness of \(A-B\). For $x \in \mathbb{R}^n$, we let $\|x\|$ denote the Euclidean norm and $\|x\|_{M} := \sqrt{x^{\mathsf{T}} Mx}$ to denote the weighted norm with respect to a positive definite matrix \(M\). For a matrix \(A\in \Rbb^{l\times n}\), \(A_i\) denotes the \(i\)th row of \(A\) and for a vector \(b\in \Rbb^n\), \(b_i\) denotes the \(i\)th element of \(b\). For two sets \(A\) and \(B\), \(A\oplus B = \{a + b \mid a\in A, b\in B\}\) denotes the Minkowski set addition over the sets. The smallest eigen value of a symmetric matrix \(M\) is denoted by \(\lambda_{\text{min}}(M)\). We use the bar accent \((\bar{\mathbf{x}}, \bar{\mathbf{u}})\) to denote nominal dynamics and hat notation \(\hat{a}\) to denote learned components. 
\label{notation}

\section{Problem Setup \& Preliminaries}

We consider a class of discrete-time nonlinear systems with additive stochastic process noise, given by,
\begin{align}\label{eq:sys_dyn}
    x_{k+1} &= f(x_k, u_k) + D(x_k) w_k
\end{align}
where  \(x_k\in \Rbb^{n_x} \), \(u_k\in \Rbb^{n_u}\) and \(w_k \in \Rbb^{n_w}\)  are the state, controls and process noise at time \(k\in\Ibb_{\geq0}\), respectively. Here \(f\) is a known smooth function and \(D:\Rbb^{n_x}\rightarrow \Rbb^{n_x \times n_w}\) is a known state-dependent matrix function.  

\begin{assumption}\label{ass:iid_noise}
    The process noise \(\mathbf{w}\) consists of i.i.d samples from an unknown zero mean probability distribution \(\mathcal{Q}\), that is, \(w_k\sim \mathcal{Q}\) and \(\Ex[w_k]=0\) for all \(k\in \Ibb_{\geq0}\). A disturbance dataset is available, given by \(\D_w = \{\mathbf{w}^{\scriptscriptstyle(0)},\mathbf{w}^{\scriptscriptstyle(1)}, \ldots , \mathbf{w}^{\scriptscriptstyle(K-1)}\}\) for \(K\in \Nbb\), where each \(\mathbf{w}^{\scriptscriptstyle(j)} \) represents samples of process noise over the time horizon \(N\), with \(w^{\scriptscriptstyle(j)}_k \sim \mathcal{Q} \) for all \(k\in \Ibb_{[0,N-1]}\). The support of distribution \(\mathcal{Q}\) can be unbounded.
\end{assumption}

Given an initial condition \(x_0 = \mathbf{x}(0)\), we aim to drive the system~\eqref{eq:sys_dyn} to a target region \(\mathcal{X}_N\subset \Rbb^{n_x}\) over a time horizon \(N\in \mathbb{N}\) , while satisfying state constraints \(\mathcal{X}\subset \Rbb^{n_x}\) and minimizing given step cost \(c(\cdot)\) and final cost \(c_F(\cdot)\). This problem can be formulated as a chance-constrained optimal control problem, given by, 

\begin{problem}
\label{prob:problem1}
\begin{subequations}
\begin{align}
\min_{\mathbf{x},\mathbf{u}}\;&  \;\;
   \Ex\!\left[ c_F(x_N) + \sum_{k=0}^{N-1} c(x_k, u_k) \right] \label{eq: SOCP_Cost} \\
\text{s.t.}\;& x_{k+1} = f(x_k, u_k) + D(x_k)w_k && k\in \Ibb_{[0,N-1]} \nonumber \\
& \Pr(x_k \in \X) \ge 1-p &&  k\in \Ibb_{[1,N-1]} \label{eq:SOCP_stateconstraints}\\
& \Pr(x_N \in \X_N) \ge 1-p \label{eq:SOCP_goalconstraint}\\
& x_0 = \mathbf{x}(0) \nonumber 
\end{align}
\end{subequations}
\end{problem}
where \(p\in(0,1)\) denotes the maximum allowed probability of constraint violation. Problem~\ref{prob:problem1} is generally intractable due to nonconvex distributional constraints on the uncertain states. A tractable deterministic reformulation of the problem becomes challenging since the distribution \(\mathcal{Q}\) is not explicitly known, on top of the difficulties in propagating this uncertainty through the nonlinear dynamics. To tackle this, we will first restrict the controls to time-parametrized state feedback policy \(u_k=\pi_k(x_k)\) and then consider chance constraint satisfaction for the closed-loop dynamics, 
\begin{align}\label{eq:CL_dyn}
    x_{k+1} = \varphi_\pi(x_k,k) + D(x_k)w_k
\end{align}
where \(\varphi_\pi(x, k) \coloneqq f(x, \pi_k(x))\) denotes the vector field of the nominal closed-loop dynamics. 

Our approach makes use of the disturbance dataset \(\D_w\) and is based on contraction theory and finite sample guarantees associated with statistical inference. It designs a suitable policy \(\pi_k(\cdot)\) which is a feasible solution to Problem~\ref{prob:problem1}; that is, the closed-loop state \(x_k\) in~\eqref{eq:CL_dyn} satisfies constraints \(x_k \in \mathcal{X}\) and \(x_N \in \mathcal{X}_N\) with a confidence of at least \((1-p)\). We now briefly introduce the necessary preliminaries. 


\textbf{Closed-loop Contraction:} Contraction theory offers a set of tools to analyze and control the incremental stability of nonlinear system trajectories with respect to each other~\cite{slotine1998contraction}.

Consider the nominal closed-loop dynamics, $x_{k+1} = \varphi_\pi (x_k, k).$ Let \(F_k \coloneqq \tfrac{\partial \varphi_\pi(x,k)}{\partial x}\) denote the Jacobian matrix of \(\varphi_\pi\), and \(\Gamma(a_k,b_k)\) denote the set of piecewise-smooth curves \(c:s\in[0,1]\rightarrow \Rbb^{n_x}\) connecting \(a_k\) to \(b_k\) for some \(a_k, b_k \in \Rbb^{n_x}\), with \(c(0)=a_k\) and \(c(1)=b_k\). For a square matrix function \(\Theta(x,k) \in \mathbb{R}^{n\times n}\) and uniformly positive definite matrix function  \(M(x,k) \coloneqq \Theta(x,k)^\top \Theta(x,k) \in \mathbb{R}^{n\times n}\), the Riemann energy for defining the distance between \(a_{k}\) and \(b_{k}\) is given by~\cite{ccm}, 
\begin{align}
\label{eq:REnergy}
    E(a_{k}, b_{k}) \coloneqq \underset{c\in\Gamma(a_{k}, b_{k}) }{\min} \int_0^1 \frac{\partial c(s)}{\partial s}^\top  M(c(s), k)\frac{\partial c(s)}{\partial s} ds.
\end{align} 
Lemma~\ref{lemma:DTContraction} establishes conditions for uniform incremental stability for the nominal closed-loop system~\cite{slotine1998contraction,wei2021discretetime}.
 
\begin{lemma}\label{lemma:DTContraction}
    Consider a uniformly bounded matrix function \(M_k \coloneqq M(x,k)\), i.e., \(\underline{m}I \preceq M_k \preceq \overline{m}I\) for some \(\overline{m} \geq \underline{m} >0\). If \(F_k^\top M_{k+1} F_k \preceq \lambda M_k \) with \(\lambda\in [0,1)\) holds for all \(x\in\Rbb^{n_x}\) and \(k\in \Ibb_{\geq0}\), then the energy \(E(x_{1,k}, x_{2,k})\) decreases exponentially over the flow of the nominal closed-loop system and satisfies the decay condition \(E(\varphi_\pi(x_{1,k},k), \varphi_\pi(x_{2,k},k)) \leq \lambda E(x_{1,k},x_{2,k})\), where $(x_{1,k}, x_{2,k})$ denotes any pair of the solution trajectories. The nominal closed-loop system is then said to be globally contracting with rate \(\lambda\).
\end{lemma}

\textbf{Challenges in Metric Synthesis:} Control Contraction Metric (CCM)~\cite{ccm} establishes conditions for the existence of a feedback tracking controller such that the closed-loop system is contracting with a given rate. The search for a suitable CCM and corresponding differential feedback gain requires solving infinite-dimensional linear matrix inequalities, which are often solved using finite-dimensional approximations with restrictive assumptions, e.g., sum-of-squares relaxation~\cite{ccm, SinghRobustFeedbackPlanning}. A tracking controller is then constructed by computing the geodesics with respect to the calculated CCM. Existing methods propose neural network parameterizations of a contraction metric and its tracking controller~\cite{tsukamoto2020neural,SunJhaFanLearningCertifiedCCM}. However, the quality of such learned components depends on their structural priors and training, and thus remains primarily empirical. Verifying the validity, in terms of satisfying conditions for incremental stability, becomes crucial for their trustworthy usage.

\textbf{Conformal Prediction (CP):} CP is a statistical tool for uncertainty quantification that leverages rank statistics of \emph{exchangeable} random variables. It uses nonconformity scores to construct prediction intervals that are valid up to the user-defined probability level while being distribution agnostic~\cite{shafer2008tutorial, lindemann2024formal}. CP is also model-agnostic in nature and holds out when the nonconformity scores result from a function composition, e.g., a neural network. Let \(S^{\scriptscriptstyle(1)}, \dots, S^{\scriptscriptstyle(K)}, S^{\scriptscriptstyle(K+1)}\) be a collection of \((K+1)\) nonconformity scores, where \( \mathcal{S} = \{S^{\scriptscriptstyle(1)}, \dots, S^{\scriptscriptstyle(K)}\}\) are scores drawn from the calibration distribution and \(S^{\scriptscriptstyle(K+1)}\) corresponds to the test sample. In robotic applications, the assumption of exchangeability can be too restrictive for inference on the test sample, due to possible distribution shifts with respect to the calibration dataset. This paper thus employs a CP variant called Weighted Conformal Prediction (W-CP)~\cite{WCP_Barberetal}, which relaxes exchangeability by reweighting calibration samples to better match the test distribution. Let \(\{\tilde{w}_j\}_{j=1}^K\) be a set of positive weights with \(\tilde{w}_j \in [0,1]\) for all \(j\in \Ibb_{[1,K]}\). Define the weighted quantile as, 
\begin{align}
    q_{1-\alpha}(\mathcal{S}) \coloneqq \text{Quantile}_{1-\alpha}\left(\sum_{i=1}^{K} \bar{w}_i\delta_{S^{\scriptscriptstyle (i)}} + \bar{w}_\infty\delta_\infty\right) \nonumber    
\end{align}
where \(\bar{w}_j = \tilde{w}_j /(\sum_{i=1}^{K} \tilde{w}_i + 1)\),  \(\bar{w}_\infty = 1 / (\sum_{i=1}^{K} \tilde{w}_i + 1)\) and \(\delta_a\) denotes the point mass probability at \(a\in \mathbb{R}\cup \{-\infty, +\infty\}\). 

\begin{lemma} \label{lemma:CPlemma}
    \cite{WCP_Barberetal} Let \(\mathsf{d}_{\text{TV}} (\mathcal{P}_1, \mathcal{P}_2)\) be the total variation distance between distributions  \(\mathcal{P}_1\) and \(\mathcal{P}_2\). For any \(\alpha\in (0,1)\), 
    \begin{align}\label{eq:WCP_guarantee}
        \Pr \left(S^{\scriptscriptstyle(K+1)} \leq q_{1-\alpha}(\mathcal{S}) \right) \geq 1-\bar{\alpha}
    \end{align}
    where \(\bar{\alpha} = \alpha + \sum_{i=1}^{K} \bar{w}_i \mathsf{d}_{\text{TV}} (\bar{\mathcal{S}}, \bar{\mathcal{S}}_i)\). Here \( \bar{\mathcal{S}} = \mathcal{S}\cup S^{\scriptscriptstyle (K+1)}\) and \(\bar{\mathcal{S}}_i\) is the set \(\bar{\mathcal{S}}\) with scores \(S^{\scriptscriptstyle(i)}\) and \(S^{\scriptscriptstyle(K+1)}\) swapped for each $i = 1,\cdots,K$, and \(\Pr\) denotes the marginal probability measure over randomness in draws of the scores in \(\mathcal{S}\). 
\end{lemma} 
\begin{remark}\label{remark:CPremark}
    The vanilla CP guarantee from~\cite{shafer2008tutorial} is obtained for exchangeable scores by setting \(\tilde{w}_1 =\cdots= \tilde{w}_K = 1\). Assuming \(\bar{w}_1S^{\scriptscriptstyle(1)}\leq \cdots \leq \bar{w}_K S^{\scriptscriptstyle(K)}\), if \(K \geq \lceil (1-\alpha)(K+1)\rceil \), one can pick \(q_{1-\alpha}(\mathcal{S}) = \bar{w}_{j_\alpha} S^{\scriptscriptstyle(j_\alpha)}\) where \(j_{\alpha} = \lceil (1-\alpha)(K+1)\rceil\). The weights are chosen so as to minimize the coverage gap \(\sum_{i=1}^{K} \bar{w}_i \mathsf{d}_{\text{TV}} (\mathcal{S}, \mathcal{S}_i)\) with higher weights assigned to data points that resemble the test distribution more closely~\cite{WCP_Barberetal, confcont3}.
\end{remark}



\section{Uncertainty Quantification \& Contraction}
\label{sec:UQ_sec}
We aim to construct the following confidence sets for the closed-loop system~\eqref{eq:CL_dyn} under a suitable policy \(\pi(\cdot)\), with finite samples of the noise distribution \(\D_w\).
\begin{definition}\label{def:confidence_set_def}
    Let \(x\in\Rbb^{n_x}\) be a random vector and \(M\) be a positive definite matrix. For \(\mu\in \Rbb^{n_x}\), \(\mathcal{B}^\theta (\mu, M) \coloneqq \left\{ \tilde{x} \mid \|\tilde{x}- \mu\|_{M^{-1}} \leq 1 \right\}\) is an ellipsoidal confidence set of probability level \(\theta\in(0,1)\) for \(x\) if \(\Pr\left(x\in \mathcal{B}^\theta(\mu,M)\right) \geq \theta\). 
\end{definition}

We call the state and control trajectory pair \((\mathbf{\bar{x}}, \mathbf{\bar{u}})\) as a target trajectory if they are forward complete solutions of the nominal dynamics \(\bar{x}_{k+1} = f(\bar{x}_k, \bar{u}_k)\) in~\eqref{eq:sys_dyn}. Let us consider neural network predictions for a state-independent CCM \(\hat{M}_k\) and the tracking policy \(\hat{\pi}(\cdot,\bar{x}, \bar{u})\), trained to make the nominal closed-loop system controlled by \(u_k = \hat{\pi}(x_k, \bar{x}_k, \bar{u}_k)\) contracting with rate \(\lambda\in [0,1)\) under the metric defined through \(\hat{M}_k\).  We will consider the target tracking feedback policy \(u_k = \pi_k(x_k) = \hat{\pi}(x_k, \bar{x}_k, \bar{u}_k)\), thereby reformulating Problem 1 into a deterministic search over the space of target trajectories \((\mathbf{\bar{x}}, \mathbf{\bar{u}})\). Here, we aim to quantify the epistemic uncertainty associated with these learned components and their contribution to the uncertainty of the closed-loop dynamics, along with the effect of stochastic noise during deployment. We will perform a calibration step to evaluate the controller \(\hat{\pi}\) over random tracking scenarios in the presence of noise characterized by the disturbance dataset \(\D_w\). This will be used to show that for a suitable choice of nonconformity score, we can obtain meaningful predictive inference on the confidence set of the closed-loop state in~\eqref{prob:problem1} over the time horizon via Lemma~\ref{lemma:CPlemma}. First, we make the following assumption, 
\begin{assumption}\label{ass:feasible_policy}
    \(\hat{\pi}(x,x,u) = u \;\; \forall \;\; x\in \Rbb^{n_x}\)  and \(u\in \Rbb^{n_u}\).
\end{assumption}

\noindent This assumption ensures that the target trajectory \((\mathbf{\bar{x}}, \mathbf{\bar{u}})\) is a valid solution to the nominal closed loop system. Let \(\hat{M}_k = \Theta_k^\top \Theta_k\) and define, 
\begin{multline}
\Delta_V(x_k,\bar x_k,\bar u_k)
\coloneqq \max\Big\{0,\;
     -\lambda \Big\| \hat{\Theta}_k(x_k-\bar{x}_k)\Big\|  \\
    + \Big\| \hat{\Theta}_{k+1}
       \big(\varphi_{\hat{\pi}}(x_k)-\varphi_{\hat{\pi}}(\bar{x}_k)\big)\Big\| 
\Big\}
\label{eq:LME}
\end{multline}
where \(\varphi_{\hat{\pi}}(x_k) \coloneqq f(x_k, \hat{\pi}(x_k, \bar{x}_k, \bar{u}_k))\). Note that under a state-independent and time indexed metric defined through \(\hat{M}_k\), the geodesic is equivalent to a straight line connecting the two points and the Riemann energy \( E(x_k, \bar{x}_k)\) in~\eqref{eq:REnergy} reduces to to the weighted norm \(\|x_k-\bar{x}_k\|_{\hat{M}_k}\). Hence, given \((\bar{x}_k,\bar{u}_k)\), equation~\eqref{eq:LME} denotes the residuals of the decay condition in~\eqref{lemma:DTContraction} at time \(k\). Let \(\Omega\) denote the distribution of \(N\) length state and control target trajectories with initial state \(\bar{x}_0 = \mathbf{x}(0)\). By taking a random sample \((\mathbf{\bar{x}}^{\scriptscriptstyle(j)},\mathbf{\bar{u}}^{\scriptscriptstyle(j)}) \sim \Omega\) for every \(j\)-th sample \(\mathbf{w}^{\scriptscriptstyle(j)}\) from dataset \(\mathcal{D}_w\), one can generate \(K\) distinct closed-loop realizations \(\mathbf{x}^{\scriptscriptstyle(j)}\) of the system~\eqref{eq:sys_dyn} under the feedback policy \(u_k = \hat{\pi}(x_k, \bar{x}^{\scriptscriptstyle(j)}_k, \bar{u}^{\scriptscriptstyle(j)}_k)\), i.e., 
\begin{align}
    x^{\scriptscriptstyle(j)}_{k+1} = f\left(x^{\scriptscriptstyle(j)}_k, \hat{\pi}\left(x^{\scriptscriptstyle(j)}_k, \bar{x}^{\scriptscriptstyle(j)}_k, \bar{u}^{\scriptscriptstyle(j)}_k\right)\right) + D(x^j_k) w^{\scriptscriptstyle(j)}_k
\end{align}
with initial state \(x^{\scriptscriptstyle(j)}_0 = \mathbf{x}(0)\) for all \(j\in\Ibb_{[1,K]} \).
Now for each \(j\in \Ibb_{[1,K]}\) and \(k\in\Ibb_{[1,N]}\), define the nonconformity score as follows, 
\begin{align}\label{eq:score_def}
    S^{\scriptscriptstyle(j)}_k \coloneqq \sum_{i=0}^{k-1} \lambda^i \left( \Delta_V (x^{\scriptscriptstyle(j)}_i, \bar{x}^{\scriptscriptstyle(j)}_i, \bar{u}^{\scriptscriptstyle(j)}_i )+  \left\| \hat{\Theta}_{i+1} D(x^{\scriptscriptstyle(j)}_i) w^{\scriptscriptstyle(j)}_i\right\| \right)  ~~~~~~~~~~
\end{align}
 and collect the scores \(\mathcal{S}_k \coloneqq \{S^{\scriptscriptstyle(j)}_k\}_{j=1}^K\). This will be our calibration dataset at time \(k\). Note that the nonconformity score defined in equation~\eqref{eq:score_def} maps a target trajectory \((\mathbf{\bar{x}}^{\scriptscriptstyle(j)},\mathbf{\bar{u}}^{\scriptscriptstyle(j)})\) and noise realization \(\mathbf{w}^{\scriptscriptstyle(j)}\) to a positive real number. Hence, the calibration distribution can be identified with \((\Omega, \mathcal{Q}^N)\). Let \(u_k = \hat{\pi}(x_k, \bar{x}^{\scriptscriptstyle(K+1)}_k, \bar{u}^{\scriptscriptstyle(K+1)}_k)\) denote the control policy during deployment. The corresponding target trajectory \((\mathbf{\bar{x}}^{\scriptscriptstyle(K+1)},\mathbf{\bar{u}}^{\scriptscriptstyle(K+1)}) \sim \Xi\) is obtained by solving Problem 1, as to be described in Section~\ref{sec:ConstraintTightening}, and not known apriori. Here, \(\Xi\) denotes the distribution of target trajectories that are feasible solutions to the proposed optimization problem. We propose the following result for the test score $S^{\scriptscriptstyle(K+1)}_k$ obtained from the test distribution \((\Xi, \mathcal{Q}^N)\).
 

\begin{theorem}\label{thrm:confset_thrm}
    Suppose that the Assumptions~\ref{ass:iid_noise} and~\ref{ass:feasible_policy} hold. Consider a sample target trajectory \( (\bar{\mathbf{x}},\bar{\mathbf{u}}) \sim \Xi\). Let \(\bar{w}_j \) denote the normalized weights for \(j\in\Ibb_{[1,K]}\). Suppose also that the coverage gap \(\sum_{i=1}^{K+1} \bar{w}_i\mathsf{d}_{\text{TV}} (\bar{\mathcal{S}}_{k}, \bar{\mathcal{S}}_{k,i})\) is upper bounded by \(\bar{\delta}\), where \(0\leq\bar{\delta}<\delta\) for some \(\delta\in(0,1)\) and for all \(k\in \Ibb_{[1,N]}\). Then \(\mathcal{B}^{1-\delta}(\bar{x}_k, W_k)\) is a confidence set with probability level \((1-\delta)\) for the state \(x_k\) in~\eqref{eq:sys_dyn} at time \(k\in \mathbb{I}_{[1,N]}\) under the feedback policy  \(u_k = \hat{\pi}(x_k, \bar{x}_k,\bar{u}_k)\), where \(W_k \coloneq C^2_k\hat{M}^{-1}_k\) and \(C_k \coloneqq q_{1-\delta + \bar{\delta}}(\mathcal{S}_k)\). 
\end{theorem}
\begin{proof}
    Define \(V_k\coloneq V(x_k, \bar{x}_k) = \|\hat{\Theta}_k (x_k - \bar{x}_k)\|\). Then we have, 
    \begin{align}
        &V_{k+1}  = \left\| \hat{\Theta}_{k+1}(x_{k+1} - \bar{x}_{k+1})\right\| \nonumber \\
        & \leq \left\| \hat{\Theta}_{k+1}( \varphi_{\pi}(x_k) - \bar{x}_{k+1})\right\| + \left\|\hat{\Theta}_{k+1}D(x_k) w_k\right\| \nonumber \\
        & \overset{\eqref{eq:LME}}{\leq} \lambda V_k + \Delta_V(x_k,\bar{x}_k,\bar{u}_k) +  \left\|\hat{\Theta}_{k+1}D(x_k) w_k\right\| \nonumber \\
        & \leq \lambda^{k+1} V_0 + \sum_{i=0}^{k} \left(\Delta_V(x_i,\bar{x}_i,\bar{u}_i) + \left\|\hat{\Theta}_{i+1} D(x_i)w_i\right\| \right)\lambda^i \nonumber \\
        &\overset{\eqref{eq:score_def}}{=} S^{\scriptscriptstyle(K+1)}_{k+1}.   \nonumber
    \end{align}
    Note that \(V_k = \| x_k - \bar{x}_k\|_{\hat{M}_k}\) and \(V_0 = 0\) by construction. Evaluate the quantile \(C_k= q_{1-\delta + \bar{\delta}}(\mathcal{S}_k)\) (see Remark~\ref{remark:CPremark}) for each \(k\in \Ibb_{[1,N]}\) . Applying Lemma~\ref{lemma:CPlemma} and taking the complement, we obtain, 
    \begin{align}
         \Pr \left(V_k > C_k \right) & \leq \delta - \bar{\delta} +  \sum_{i=1}^{K} \bar{w}_i \mathsf{d}_{\text{TV}} (\bar{\mathcal{S}}_{k}, \bar{\mathcal{S}}_{k,i}) \nonumber \leq \delta. \nonumber
    \end{align}
    Taking a complement again gives the desired claim.  
\end{proof}


\begin{remark}\label{remark:OnTheorem1}
   Note that  \(W_k = C^2_k\hat{M}^{-1}_k\) dictates the size and shape of the prediction set of \(x_k\), analogous to that of covariance in Gaussian-based~\cite{PRSpaperHewing} methods. Our approach, as in existing contraction-based methods (\S 4.2,\cite{SinghRobustFeedbackPlanning}) allows us to control the associated conservativeness up to a certain degree by an appropriate design choice of contraction metric bounds \(\underline{m}, \overline{m}\) and the contraction rate \(\lambda\) in equation~\eqref{eq:LME}. 
\end{remark}

\begin{remark}\label{remark:inital_state_distribution}
    Lemma~\ref{lemma:CPlemma} provides a way to handle the loss of exchangeability caused by distribution shift between simulated rollouts and real-world deployment. To generalize the setup in Lemma~\ref{lemma:CPlemma}, \(D_w\) may be constructed using samples from an approximate noise distribution \(\hat{\mathcal{Q}}\) instead of the actual noise distribution \(\mathcal{Q}\) (see Section~\ref{subsec:crazyflieHardware}). In such cases, the application of Lemma~\ref{lemma:CPlemma} and subsequently Theorem~\ref{thrm:confset_thrm} will assume the gap between the calibration distribution \((\Omega, \hat{\mathcal{Q}}^N)\) and the unknown test distribution \((\Xi, \mathcal{Q}^N)\) can be sufficiently controlled by the weights. We could also update $\delta$ of Lemma~\ref{lemma:CPlemma} using adaptive CP~\cite{2020_Tibshirani_CP-CS}. 
\end{remark}

Although we do not consider closed-loop control constraints in our approach, we can obtain such guarantees by imposing some structure on the learned controller \(\hat{\pi}\), as we show in the following result.  

\begin{corollary}\label{prop:control_conf_set}
Let \(\mathcal{B}^{1-\delta}(\bar{x}_k, W_k)\) be a confidence set for \(x_k\) at probability level \((1-\delta)\). Suppose \(u_k = \hat{\pi}(x_k,\bar{x}_k, \bar{u}_k) = k(x_k, \bar{x}_k) + \bar{u}_k\) and Assumption~\ref{ass:feasible_policy} holds, i.e., \(k(x,x) = 0\) for all \(x\in \Rbb^{n_x}\). If \(k(\cdot)\) is Lipschitz with Lipschitz constant \(L\), then \(\mathcal{B}^{1-\delta} (\bar{u}_k, Z_k)\) is a confidence set for \(u_k\) at probability level \((1-\delta)\), where \(Z_k = \tfrac{L^2}{\lambda_{\text{min}} (W_k^{-1})} I\) 
\end{corollary}
\begin{proof}
We have the following, 
\begin{align*}
    \|u_k - \bar{u}_k\| &= \|k(x_k, \bar{x}_k)\| = \|k(x_k, \bar{x}_k) - k(\bar{x}_k, \bar{x}_k)\|  \\
    & \leq L \|x_k - \bar{x}_k\| \leq L\frac{\|x_k - \bar{x}_k\|_{W_k^{-1}}}{\sqrt{\lambda_{\text{min}} (W_k^{-1})}}.
\end{align*}
From Definition~\ref{def:confidence_set_def}, we get \(\|x_k - \bar{x}_k\|_{W_k^{-1}}\leq 1\) with probability at least \((1-\delta)\), which implies, \(\Pr (\|u_k - \bar{u}_k\| \leq \frac{L}{\sqrt{\lambda_{\text{min}} (W_k^{-1})}}) \geq 1-\delta \). Then the result follows also from Definition~\ref{def:confidence_set_def}. 
\end{proof}

\section{Deterministic Reformulation}\label{sec:ConstraintTightening}

We will use confidence sets constructed using Theorem~\ref{thrm:confset_thrm} to reformulate chance-constraints in~\eqref{eq:SOCP_stateconstraints},~\eqref{eq:SOCP_goalconstraint} into deterministic constraints on the target trajectory \((\bar{\mathbf{x}},\bar{\mathbf{u}})\) via suitable constraint tightening. Essentially, we aim to find a tightened state constraint set \(\mathcal{X}^c_k\) such that \(\mathcal{X}^c_k \oplus \mathcal{B}^{1-\delta}(\bar{x}_k, W_k)\subseteq \mathcal{X} \) for all \(k\in \Ibb_{[0,N]}\), whereby restricting \(\bar{x}_k \in \mathcal{X}^c_k\) ensures \(x_k \in \mathcal{X}\) with probability at least \((1-\delta)\) . We will use the approach from~\cite{LewCCSCP} to construct suitable convex over approximations of the set \(\mathcal{X}^c_k\) for particular classes of sets \(\mathcal{X}\). 

Let \(\mathcal{X}\) be defined by polytopic and non-convex obstacle avoidance constraints on the state, \(\mathcal{X} = \mathcal{X}_{\text{poly}} \cap \mathcal{X}_{\text{free}}\), where \(\mathcal{X}_{\text{poly}} = \{x\in \Rbb^{n_x} \;\; \vert \;\; Ax\leq b\}\), \(A\in \Rbb^{l\times {n_x}}\), \(b\in\Rbb^l\), and \(\mathcal{X}_{\text{free}} = \Rbb^{n_x} \backslash \cup_{i=1}^J \mathcal{O}_i\). Here \(\mathcal{O}_i \subset \Rbb^{n_x}, i\in \Ibb_{[1,J]}\) are closed convex sets representing \(J\) obstacles. First, we split the joint chance constraints \( \Pr(x_k\in \mathcal{X}_{\text{poly}} \cap \mathcal{X}_{\text{free}})\geq 1-p\) into single chance-constraints using the Bonferroni inequality~\cite{NakkaCCnonlinearS} with equal risk allocation accordingly, 
\begin{align}
    \Pr(x_k\in \mathcal{X}_{\text{poly}})\geq 1-\bar{p} \label{eq:eqriskPolyCC}\\
    \Pr(x_k\in \mathcal{X}_{\text{free}})\geq 1-\bar{p} \label{eq:eqriskObsCC}
\end{align}
\vspace{1ex} 
where \(\bar{p} = p/2\). The above result is obtained by applying Boole's inequality over the complement of the event \(x_k \in \mathcal{X}_{\text{poly}} \cap \mathcal{X}_{\text{free}}\). Let \(\mathcal{B}^{1-\bar{p}}(\bar{x}_k, W_k)\) be a confidence set for \(x_k\) with probability level \((1-\bar{p})\). Then the polytopic chance constraint may be conservatively reformulated as \(l\) deterministic constraints~\cite{LewCCSCP}, 
\begin{align}
    \eqref{eq:eqriskPolyCC} \;\; \Leftarrow \;\; A_i\bar{x}_k + \sqrt{A_i W_k A_i^{\top}}  \leq b_i \;\; i\in \Ibb_{[1,l]}.\label{eq:polyConstraintTight}
\end{align}
Here the inequality condition in~\eqref{eq:polyConstraintTight} is the robust counter part to the condition \(A_ix_k\leq b_i\) over the confidence set \(\mathcal{B}^{1-\bar{p}}(\hat{x}_k, W_k)\)~\cite{Koller_etal}. For the obstacle avoidance constraints, we will use the signed distance function \(d_i:\mathcal{X}\rightarrow \Rbb\) that returns the shortest distance from point \(x_k\) to the boundary \(\partial \mathcal{O}_i\) of the set \(\mathcal{O}_i\)~\cite{LewCCSCP}. It is defined as, 
\begin{align}
    d_i(x_k) \coloneqq \inf_{y\in \mathcal{O}_i} \|x_k-y\| - \inf_{z\notin \mathcal{O}_i} \|x_k-z\|
\end{align}
where \(d_i(x_k)\geq0 \implies x_k \notin \mathcal{O}_i\). Define, 
\begin{align}
    n_i(x) = \frac{x-\breve{x}_i}{d_i(x)} , \;\; \breve{x}_i = \argmin_{y\in \partial \mathcal{O}_i} \|x-y\|.
\end{align}
Then, we can conservatively reformulate obstacle avoidance chance constraint~\eqref{eq:eqriskObsCC} into \(J\) deterministic constraints as, 
\begin{align}
    \eqref{eq:eqriskObsCC} \;\; \Leftarrow \;\; d_i(\bar{x}_k) - \sqrt{n_i^\top W_kn_i} \geq 0, \;\; i\in \Ibb_{[1,J]}.\label{eq:ObsConstraintTight}
\end{align}
For more details, refer~\cite{LewCCSCP}. Let \(\mathcal{X}_N = \{x \mid Hx\leq h, H\in \Rbb^{r\times n_x}, h\in \Rbb^{r}\}\). The deterministic reformulation of~\eqref{eq:SOCP_goalconstraint} then follows the same approach as in~\eqref{eq:polyConstraintTight} with a confidence set of probability level \((1-p)\), which here gives \(r\) deterministic constraints.

Now consider the following relaxed deterministic reformulation of Problem~\ref{prob:problem1} for some \(\delta\in (0,1)\).
\begin{problem}
\label{prob:problem2}
\begin{subequations}
\begin{align}
\min_{\bar{\mathbf{x}},\bar{\mathbf{u}}}\;&  \;\;  c_F(\bar{x}_N) + \sum_{k=0}^{N-1} c(\bar{x}_k, \bar{u}_k)  \\
\text{s.t.}\;& \bar{x}_{k+1} = f(\bar{x}_k,\bar{u}_k), \;\;\;\; k\in \Ibb_{[0,N-1]}  \nonumber\\
& A_i\bar{x}_k + \bar{\eta}_k(\delta)\|A_i\|_{M^{-1}_k}  \leq b_i, \;\; i\in \Ibb_{[1,l]},  k\in \Ibb_{[1,N-1]} \nonumber\\
& d_i(\bar{x}_k) - \bar{\eta}_k(\delta)\|n_i\|_{M^{-1}_k} \geq 0, \;\;i\in \Ibb_{[1,J]}, k\in \Ibb_{[1,N-1]}  \nonumber \\
& H_i\bar{x}_N + \eta (\delta) \|H_i\|_{M^{-1}_N}  \leq h_i, \;\; i\in \Ibb_{[1,r]} \nonumber\\
& \bar{x}_0 = \mathbf{x}(0) \nonumber
\end{align}
\end{subequations}
\end{problem}
where, \(\bar{\eta}_k (\delta)= q_{1-\tfrac{\delta}{2}}(\mathcal{S}_k)\) and \(\eta(\delta)= q_{1-\delta}(\mathcal{S}_N)\). 
We then have the following result, 


\begin{theorem}\label{thrm:main_theorem}
    Let \((\mathbf{x}^\ast, \mathbf{u}^\ast)\) be the solution to Problem~\ref{prob:problem2} for \(\delta \in (0,1)\), if it exists. Suppose that Assumptions~\ref{ass:iid_noise}~and~\ref{ass:feasible_policy} hold and let the coverage gap be upper bounded  by \(\bar{\delta}\) for all \(k\in \Ibb_{[0,N]}\). Then the control policy \(u_k = \hat{\pi}(x_k, x^\ast_{k}, u^\ast_{k})\) is a feasible solution to Problem~\ref{prob:problem1} if \(\delta = p- 2\bar{\delta}\). 
\end{theorem}
\begin{proof}
    For each \(k\in \Ibb_{[1,N-1]}\), set \(W_k = \bar{\eta}_k(\delta)^2 M^{-1}_k\) and \(W_N = \eta(\delta)^2M^{-1}_N\). From Theorem~\ref{thrm:confset_thrm}, \(\mathcal{B}^{1-\tfrac{\delta}{2}-\bar{\delta}}(x^\ast_k, W_k)\) and \(\mathcal{B}^{1-\delta-\bar{\delta}}(x^\ast_N, W_N)\) are confidence sets for \(x_k\), \(k\in \Ibb_{[0,N-1]}\), and \(x_N\) in system~\eqref{eq:sys_dyn}, respectively, under the policy \(u_k = \hat{\pi}(x_k, x^\ast_k, u^\ast_k)\). Since \(\mathbf{x}^\ast, \mathbf{u}^\ast\) satisfies the tightened constraints in Problem~\ref{prob:problem2}, from~\eqref{eq:polyConstraintTight} and~\eqref{eq:ObsConstraintTight},  we get \(\Pr(x_k\in \mathcal{X})\geq 1- \delta - 2\bar{\delta} = 1-p\) thus satisfying constraint~\eqref{eq:SOCP_stateconstraints}. Similarly, we obtain \(\Pr(x_N\in \mathcal{X}_N)\geq 1 - \delta - \bar{\delta} = 1-p + \bar{\delta} \geq 1-p\), which satisfies~\eqref{eq:SOCP_goalconstraint}.     
\end{proof}

We note that Theorem~\ref{thrm:main_theorem} quantifies the effect of stochastic noise, the quality of predicted contraction metrics and controller, and the conservativeness demanded by the coverage gap between the calibration and test distributions, in one unified way. Under the assumption of the coverage gap bound \(\bar{\delta}\), it provides a suitable method to ensure chance-constraint satisfaction up to the desired probability level when using the learned controller \(\hat{\pi}\). In Section~\ref{sec:experiments}, we validate the same through both numerical and hardware experiments. We note that, practically, it is not possible to know the coverage gap bound \(\bar{\delta}\) a priori, making our guarantees of the \textit{a posteriori} type \cite{CPforChanceConstrainedProgramming}. To address this, one may update the upper bound by means of verification or re-estimate the quantile online using adaptive CP~\cite{2024_Zhou_ACP-CBF-MPC}.

\section{Case Studies}
\label{sec:experiments}

We present our approach with the Dubins car model in numerical simulation and with the Crazyflie quadcopter in the real world. The code for all experiments is made available at \myhref{https://github.com/Rihan24/SCC-TrajOpt}{https://github.com/Rihan24/SCC-TrajOpt}.
In both experiments, we show that the confidence sets constructed using our approach from finite noise samples are valid up to the prescribed probability level, and the proposed chance constraints are satisfied. We remark that the main goal of this section is the empirical validation of the formal guarantees established by our theorems, rather than demonstrating the superiority of our approach over all existing data-driven/learning-based motion planning and control methods. For example, even in the case of a poorly learned control policy with a suboptimally designed contraction metric, our theorems still provide a means to formally quantify their performance without relying on structural or distributional priors. For this reason, we compare against only a limited set of existing approaches in the following. These serve solely to illustrate our claims within the constraints of the page limit.

We leverage a neural network parameterized CCM \(M(x;\theta_m)\) and its controller \(\pi(x, x^\ast, u^\ast; \theta_u)\) from the implementation in~\cite{SunJhaFanLearningCertifiedCCM}, which satisfies Assumption~\ref{ass:feasible_policy} by design. We consider \(\hat{M}_k = M=  \hat{M}(\mathbf{x}(0);\theta_m)\) as the predicted matrices for approximately constructing contraction metrics for all \(k\in \Ibb_{[1,N]}\). The discrete-time contraction rate \(\lambda\in [0,1)\) is obtained from the chosen continuous time counterpart \(\gamma>0\) in~\cite{SunJhaFanLearningCertifiedCCM} using \(\lambda =\sqrt{(1- 2\gamma\Delta t)(\underline{m}/ \overline{m})}\)~\cite{tsukamoto2021contraction}, where \(\Delta t\) denotes the sampling rate for the discrete-time system. We note that the quantiles \(\bar{\eta}_k (\delta)\) quickly converge to \(\eta(\delta)\) in a few time steps in both experiments for our choices of \(\lambda\). Hence, we conservatively set \(\bar{\eta}_k (\delta) = \eta(\delta)\) for all \(k\in \Ibb_{[1, N-1]}\), which, along with the choice of constant contraction metric, results in confidence sets of the same size over the time horizon. We consider step cost and final cost as \(c(x_k,u_k) = u_k^\top Ru_k\) and \(c_F =0\), respectively, where \(R\) is some positive definite weight matrix. The optimization in Problem~\ref{prob:problem2} is posed and solved using the CasADi interface~\cite{casadiPaper} with the NLP solver IPOPT~\cite{IpoptPaper}. 

\textbf{Choice of \(\boldsymbol{\Omega}\):} For both experiments, we follow the approach from~\cite{SunJhaFanLearningCertifiedCCM} to sample suitable target trajectories. We first parametrize the target control trajectory \(\mathbf{\bar{u}}\) as a linear combination of fixed sinusoidal basis functions with a predefined set of frequencies. We then sample i.i.d. reference control trajectories by randomly sampling the weight for each frequency component from a normal distribution. The initial state \(\bar{x}_0\) is set to \(\mathbf{x}(0)\). We then generate the corresponding target state trajectory \(\mathbf{\bar{x}}\) by propagating  \(\bar{x}_0\) through the nominal dynamics \(\bar{x}_{k+1} = f(\bar{x}_k, \bar{u}_k)\). 

In both simulation and hardware experiments, the distribution shift between the calibrated and test distributions was minimal and well-accommodated by the slight conservativeness associated with our method. Therefore, using vanilla CP with \(\delta = p\) gave us satisfactory results. For online motion planning, one may consider receding horizon alternatives of our approach and experiments.

\subsection{Numerical Simulations}\label{subsec:dubinscar}
We consider the Dubin's car model from~\cite{SunJhaFanLearningCertifiedCCM} with 4 state variables \(x:=[p_x, p_y, \theta,v]\) where \(p_x, p_y\) are the positions, \(\theta\) is the heading angle and \(v\) is the velocity. The control inputs are \(u:=[\omega, a]\), where \(\omega\) and \(a\) are angular velocity and linear acceleration respectively. The dynamics is given by \(\dot{x} = f(x,u) = [v \cos(\theta), v \sin(\theta), \omega,a]^{\top}\).
We discretize this via a zero-order hold on the controls and add noise \(\mathbf{w}\) to obtain the discrete-time dynamics \(x_{k+1} = x_k + \Delta t f(x_k, u_k) + Dw_k\). We will consider the following two cases,
\begin{enumerate}[label=(\roman*)]
    \item Uniform noise \(w_k \sim U[w_{\text{min}}, w_{\text{max}}]\) over the hyperbox defined by \(w_{\text{min}} = [-\sigma/5, -\sigma/5, -\sigma,-\sigma]^{\top}\) and \(w_{\text{max}} = [\sigma/5,\sigma/5, \sigma,\sigma]^{\top}\). We set \(\sigma=0.15\) and \(D = \text{diag}([1,1,1,1])\). 
    \item Zero mean 3-component Gaussian mixture noise with density \(\sum_{i=1}^3 \psi_i \mathcal{N}(\mu_i, \Sigma_i) \) with \(\mu_1 = -0.5\cdot\mathds{1}_4, \mu_2 = 0\cdot\mathds{1}_4, \mu_3 = 0.5\cdot\mathds{1}_4, \Sigma_i = (0.05)^2\cdot I\) for \(i=1,2,3\), \(\psi_1, \psi_3=0.1, \psi_2=0.8\) and \(D = \text{diag}([0,0,1,1])\), where $\mathds{1}_4 = [1,1,1,1]^{\top}$. Here \(\mathcal{N}\) is the multivariate Gaussian density.
\end{enumerate}
The noise parameters are chosen so as to cater to the problem dimensions. We consider the motion planning problem from~\cite{NakkaCCnonlinearS}, with \(\mathbf{x}(0) = [0,0.4,0,0]^\top\),  goal region loosely centered at \(\mathbf{x}(N) = [10,0.4,0,0]^\top\) and collision avoidance with a wall and obstacle given by \(\mathcal{X}_{\text{poly}} = \{p_y \mid p_y\leq 2\}\) and obstacle \(\mathcal{O} = \left\{ p_{x,y} \mid \left\| p_{x,y}- p_{\text{obs}} \right\| \leq 1.2 \right\}\) respectively, where \(p_{x,y} = [p_x, p_y]^\top,p_{\text{obs}} = [0,5]^\top \). We set \(R = 0.1\cdot I\), failure probability \(\delta=p=0.1\), sampling rate \(\Delta t =0.05\) and time horizon to \(N=200\) steps. We set \(\underline{m}=0.5, \overline{m}=10\) and \(\gamma =1\), which corresponds to a discrete-time contraction rate of \(\lambda = 0.21\).  The disturbance dataset \(\D_w\) is constructed with \(K =20\) by sampling from the respective noise distributions. The quantile \(\eta(\delta)\) comes to 0.301 and 0.574 for cases (i) and (ii), respectively. We solve Problem~\ref{prob:problem2} to get \((x^\ast, u^\ast)\), which takes on average 1.67 seconds, and run 200 closed-loop realizations under policy \(u_k =\hat{\pi}(x_k, x^\ast_k, u^\ast_k)\), as shown in Figure~\ref{fig:CCM_dubins}, to verify satisfaction of chance constraints. We note that all chance constraints are satisfied at the failure probability level \(p=0.1\). In particular, the empirical failure probabilities obtained are \(\max_k \{\Pr(x_k \notin \mathcal{X})\} = 0 \%\) for case (i) and \(\max_k \{\Pr(x_k \notin \mathcal{X})\} = 1.5 \%\) for case (ii). 
\begin{figure}[htbp]
    \vspace{-1em}
    \centering
    \includegraphics[width=0.9\linewidth]{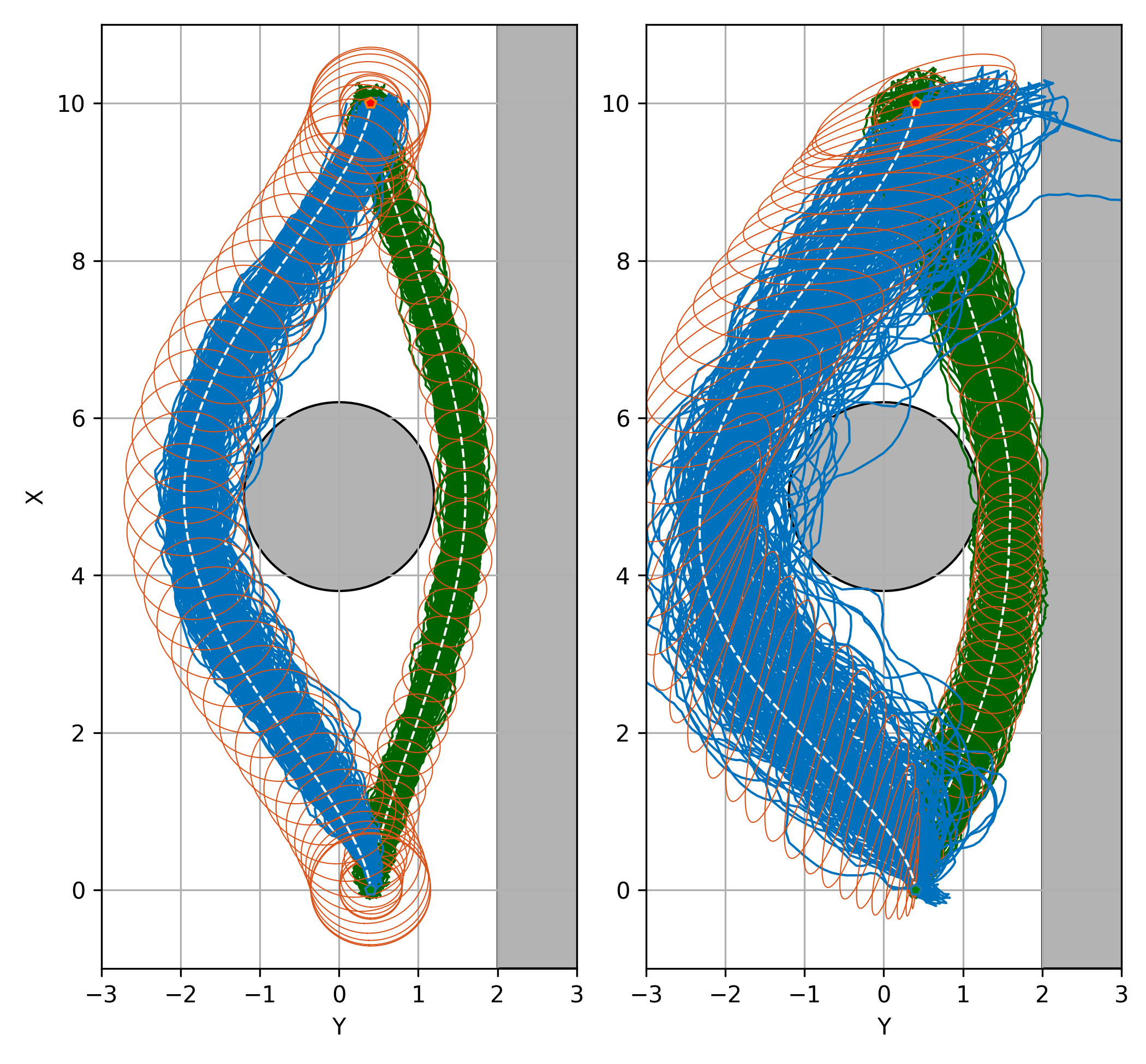}
    \vspace{-1em}
    \caption{Closed-loop realizations for the Dubin's car system with our approach (shown on the left) and linearization approach with Gaussian approximation of noise and LQR feedback (shown on the right). Case (i) Uniform noise and Case (ii) Gaussian mixture noise are shown in green and blue, respectively. The corresponding target trajectory \(\mathbf{x}^\ast\) is shown in dashed white lines, and the respective confidence sets are shown in orange.}
    \vspace{-1em}
    \label{fig:CCM_dubins}
\end{figure}

We compare our approach to the linearization baseline from ~\cite{LewCCSCP, CutiousMPCGP}, shown in Figure~\ref{fig:CCM_dubins}, with Gaussian approximation of the noise and a linear quadratic regulator (LQR) feedback policy.  The Gaussian distribution is characterized by MLE estimates of the mean and covariance from samples of the respective non-Gaussian noise distributions, and the LQR feedback gain is computed using linearized closed-loop dynamics at each state. The uncertainty of the corresponding closed-loop state at time \(k\) is characterized as a Gaussian \(x_k \sim \mathcal{N}(\mu_k, \Sigma_k)\), where \(\mu_k, \Sigma_k\) are obtained by mean and covariance propagation through the linearized dynamics. We use constraint tightening based on ellipsoidal confidence sets \(\mathcal{B}^{1-\tfrac{p}{2}}\left(\mu_k, \mathcal{X}^2_{n_x}(1-\tfrac{p}{2})\Sigma_k\right)\) for \(x_k\)~\cite{LewCCSCP}, where \(\mathcal{X}^2_{n_x}(1-\alpha)\) denotes the \((1-\alpha)\)-th quantile of the \(\mathcal{X}^2\) distribution with \(n_x\) degrees of freedom. Solving this tightened problem takes on average 103.8 seconds. From Figure~\ref{fig:CCM_dubins}, we see that Gaussian-based methods perform poorly in accurately capturing the uncertainty associated with the state over the time horizon. The empirical failure probabilities obtained are \(\max_k \{\Pr(x_k \notin \mathcal{X})\} = 10 \%\) for case (i) and \(\max_k \{\Pr(x_k \notin \mathcal{X})\} = 20.5 \%\) for case (ii). This can be attributed to heavier tail characteristics in the noise and nonlinearity in the dynamics. 



\subsection{Hardware Experiments}\label{subsec:crazyflieHardware}
We implement our approach on the \myhref{https://www.bitcraze.io/products/crazyflie-2-1-plus/}{Crazyflie 2.1} drone for safe motion planning through a static cluttered environment. We consider the quadcopter model and controller implementation from~\cite{SinghRobustFeedbackPlanning,SunJhaFanLearningCertifiedCCM}. Consider state \(x:= [p_x, p_y, p_z, \dot{p}_x,\dot{p}_y,\dot{p}_z, f, \phi, \theta, \psi]^{\top}\), where the positions \(p_x, p_y, p_z\) are expressed in the inertial world frame of the drone, \(\phi, \theta, \psi\) represent Euler angles as per the XYZ rotation sequence, both defined according to the North-East-Down frame convention. Here, \(f\) is the net mass-normalized force by the motors. Control inputs are chosen as \(u = [\dot{f}, \dot{\phi}, \dot{\theta}, \dot{\psi}]^{\top}\). The system dynamics are given by \(\dot{x} = [ \dot{p}_x,\dot{p}_y,\dot{p}_z, -f\sin(\theta), f\cos(\theta)\sin(\phi), g-f\cos(\theta)\cos(\phi), \dot{f},\dot{\phi}, \dot{\theta}, \dot{\psi}]^{\top}\), where \(g=9.81\). The contraction metric defined through \(\hat{M}\) and the corresponding tracking controller \(\hat{\pi}\) are obtained for \(\underline{m} = 0.5, \overline{m}= 25\) and \(\gamma= 0.8\). We deploy the controller offboard on the Crazyflie at 100Hz, with \myhref{https://www.qualisys.com/}{Qualisys} motion capture feeding pose estimates onto Crazyflie's internal EKF. We remark that the controller \(\hat{\pi}\) requires only a single forward pass to compute control input, and hence can even be deployed onboard the Crazyflie firmware. At each sampling instant, throttle and angle setpoints computed from \(\hat{\pi}\) are passed onto the low-level built-in PID controller onboard the Crazyflie firmware. First, thrust and angle setpoints \([f_c, \phi_c, \theta_c]\) are obtained by integrating over the control input. The respective throttle command \(\tau_c\) is computed using \(f_c\) via an iterative proportional feedback update on the previous throttle value~\cite{SinghRobustFeedbackPlanning}, which allows us to bypass the complex thrust to throttle mapping. The gain of this thrust controller is set to \(K_p=400\). The low-level controller, which is assumed to run at a higher rate, takes in throttle and angle setpoints and computes the downstream motor PWM commands. 

We consider the problem of collision-free motion planning for the Crazyflie drone over a time period of \(7.5\) seconds in a test arena cluttered with obstacles. We set the failure probability level to \(\delta = p=0.1\) and \(R = 0.01\cdot I\). The obstacles are posed as cylindrical bodies for formulating the optimization problem, as shown in Figure~\ref{fig:CCM_quad}. The yaw remains uncontrolled in our experiments. To construct the disturbance dataset \(\D_w\),  sampling from the unknown test distribution, that is, by collecting the flight data over the time horizon for  \(K \geq \lceil (1-p)(K+1)\rceil \) runs, is not practical and generalizable. Instead, we consider a Gaussian MLE estimate of the non-gaussian process noise, calculated using the formula \(w_k = x_{k+1}-f_{\text{DT}}(x_k, u_k)\) from arbitrary trial runs of the Crazyflie drone, and sample from it to generate the calibration dataset \(\D_w\) with \(K=20\).  Here \(f_{\text{DT}}\) denotes the discretized dynamics similar to that in Section~\ref{subsec:dubinscar} with \(D = I\).  We confirm that the MLE estimate is approximately zero mean. Setting \(\Delta t =0.01\) seconds, \(\lambda = 0.14\) and using approach from Sections~\ref{sec:UQ_sec}~and~\ref{sec:ConstraintTightening}, we get \(\eta(\delta)= 0.08\). We solve Problem~\ref{prob:problem2} by warm starting solver with a sample state trajectory in the obstacle free region, and deploy the control policy \(u_k=\hat{\pi}(x_k, x^\ast_k, u^\ast_k)\) over the time horizon. We run the tracking policy for 15 iterations, as shown in Figure~\ref{fig:CCM_quad}. We observe that the drone stays within the defined confidence sets throughout the time horizon in all iterations, which implies that it achieves the prescribed chance constraint. 


\begin{figure}[t]
    \centering
    \begin{subfigure}[c]{0.5\linewidth}
        \centering
        \includegraphics[width=\linewidth]{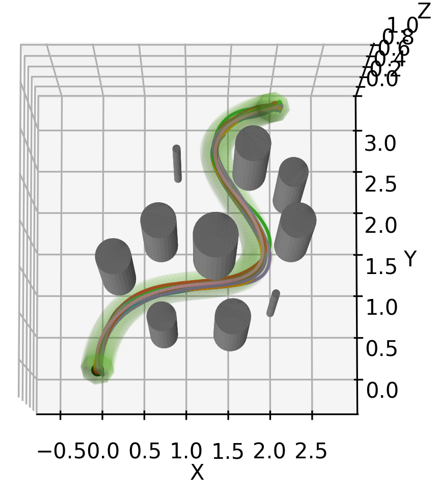}
    \end{subfigure}%
    \begin{subfigure}[c]{0.5\linewidth}
        \centering
        \includegraphics[width=\linewidth]{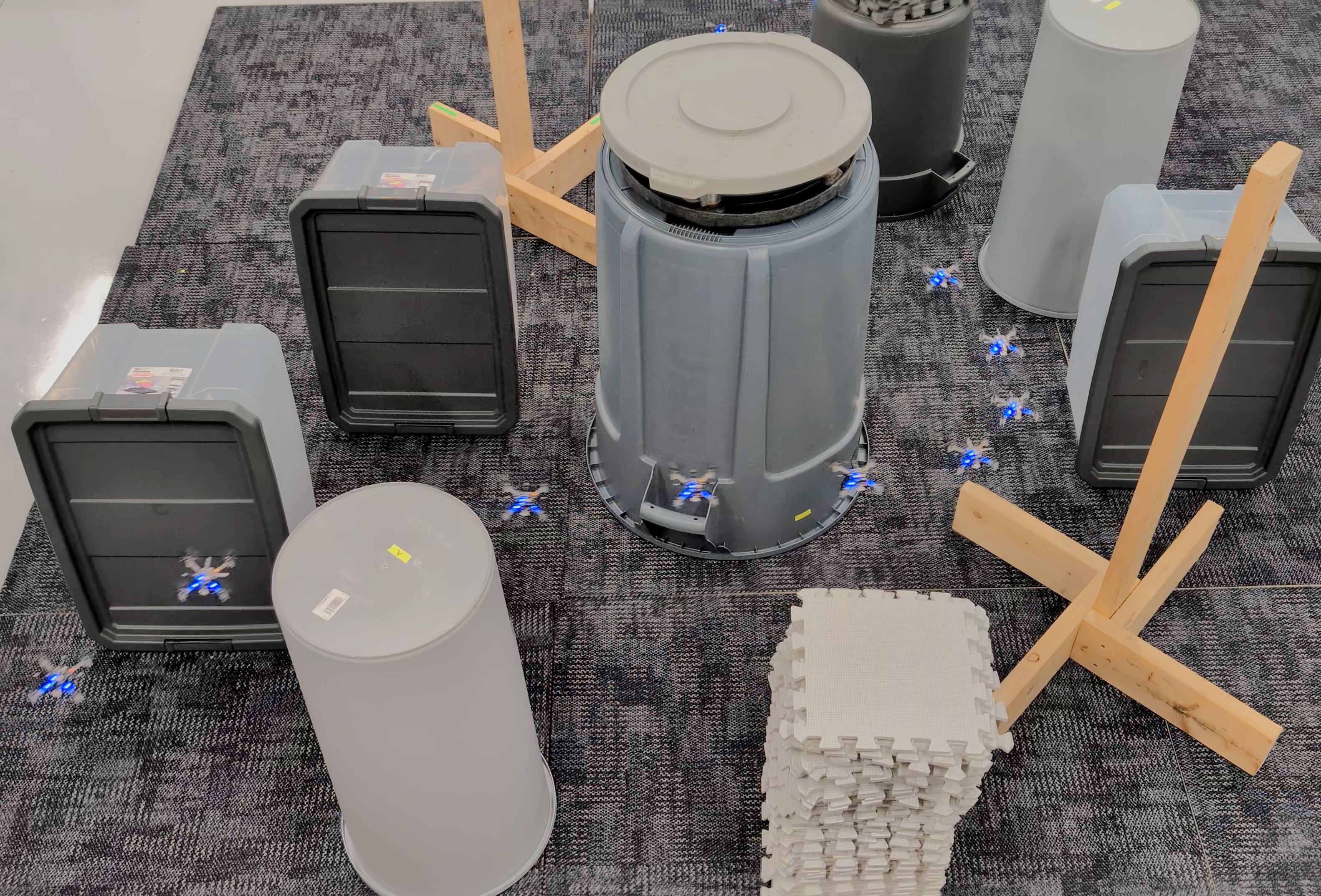}
    \end{subfigure}
    \caption{Left: Closed-loop realizations of the Crazyflie drone using our approach. The 3D confidence set \(\mathcal{B}^{1-\delta}(\hat{x}_k, \eta(\delta)^2 M^{-1})\) for the positions is shown in green. Right: One iteration of Crazyflie trajectory in the cluttered environment.}
    \label{fig:CCM_quad}
    \vspace{-2em}
\end{figure}

\section{Conclusion}
\label{sec_conclusion}
In this work, we proposed a novel method for chance-constrained trajectory optimization and control for nonlinear systems subject to non-Gaussian stochastic uncertainty. It quantifies distribution mismatch and system contraction via finite samples of the uncertainty with minimal distributional and structural assumptions. Using conformal prediction and contraction theory, we established formal guarantees for chance-constraint tightening, enabling deterministic planning with rigorous performance assurances. For example, our approach can ensure safety under non-Gaussian uncertainty even when learned representations of contraction metrics are employed, thereby bridging the gap in providing strict guarantees for learning-based motion planning.
We validated our approach with both simulation and hardware experiments, showing that it achieves the posed chance constraints up to the desired probability level, without being too conservative.

\bibliographystyle{IEEEtran}
\bibliography{bib,bibhiro}
\end{document}